\documentclass[12pt,twoside,reqno]{amsart}
\usepackage{pifont}
\usepackage{amsmath}
\usepackage{amsthm}
\usepackage{amsfonts}
\usepackage{amssymb}
\usepackage{latexsym}
\usepackage{amstext}
\usepackage{array}
\usepackage{graphicx}

\date{}
\allowdisplaybreaks[4] \footskip=15pt
\renewcommand{\uppercasenonmath}[1]{}

\newtheorem{thm}[subsection]{Theorem}
\newtheorem{cor}[subsection]{Corollary }

\newtheorem{lem}[subsection]{Lemma}

\newtheorem{prop}[subsection]{Proposition}

\usepackage{amscd,amsthm,amsfonts,amssymb,amsmath}
\usepackage[colorlinks=true]{hyperref}
\usepackage{fullpage}
\usepackage[all]{xy}

    \numberwithin{equation}{section}

\input txdtools
\begin{document}

%
%
%
%
%
%
%

\begin{center}

{\large  \bf The Parameters of Minimal Linear Codes }

\vskip 0.8cm
  {\small Wei Lu$^1$, Xia Wu$^1$\footnote{Supported by NSFC (Nos. 11971102, 11801070, 11771007)

  }}, Xiwang Cao$^2$\\

{\small $^1$School of Mathematics, Southeast University, Nanjing
210096, China}\\
{\small $^2$Department of Math, Nanjing University of Aeronautics and Astronautics, Nanjing 211100, China}\\
{\small E-mail:
luwei1010@seu.edu.cn, wuxia80@seu.edu.cn, xwcao@nuaa.edu.cn}\\
{\small $^*$Corresponding author. (Email: wuxia80@seu.edu.cn)}
\vskip 0.8cm
\end{center}

{\bf Abstract:}
Let $k\leq n$ be two positive integers and $q$ a prime power.   The basic question in minimal linear codes is to determine if there exists an  $[n,k]_q$ minimal linear code.  The first objective of this paper is to present a new sufficient and necessary condition for
 linear codes to be minimal. Using this condition, it is easy to construct minimal linear codes
or to prove  some linear codes are minimal. The second objective of this paper is to use the  new sufficient and necessary condition to partially give an answer to  the  basic question in minimal linear codes. The third objective of this paper is to present four classes of minimal linear codes, which
generalize the results about the binary case given in \cite{LY2019}. One can find that our method is
much easier and more effective.

\section{\bf Introduction}
 Let $q$ be a prime power and $\mathbb{F}_q$ the finite field with $q$ elements. Let $n$ be a positive integer and $\mathbb{F}_q^n$  the vector space with dimension $n$ over $\mathbb{F}_q$. In this paper, all vector spaces are over $\mathbb{F}_q$ and  all vectors are row vectors.  For a vector $\mathbf{v}=(v_1, \dots, v_n)\in\mathbb{F}_q^n$, let Suppt$(\mathbf{v})$ $:= \{1 \leq i\leq n : v_i\neq 0\}$ be the \emph{support} of $\mathbf{v}$. For any two vectors $\mathbf{u}, \mathbf{v}\in \mathbb{F}_q^n$, if $\rm{Suppt}(\mathbf{u})\subseteq \rm{Suppt}(\mathbf{v})$, we say that $\mathbf{v}$ covers $\mathbf{u}$ (or $\mathbf{u}$ is covered by $\mathbf{v}$) and write $\mathbf{u}\preceq\mathbf{v}$. Clearly, $a\mathbf{v}\preceq \mathbf{v}$ for all $a\in \mathbb{F}_q$.

 An $[n,k]_q$ linear code $\mathcal{C}$ over $\mathbb{F}_q$ is a $k$-dimensional subspace of $\mathbb{F}_q^n$. Vectors in $\mathcal{C}$ are called \emph{codewords}. A codeword $\mathbf{c}$ in a linear code $\mathcal{C}$ is called \emph{minimal} if $\mathbf{c}$ covers only the codewords $a\mathbf{c}$ for all $a\in \mathbb{F}_q$, but no other codewords in $\mathcal{C}$. That is to say, if a codeword $\mathbf{c}$  is minimal in  $\mathcal{C}$, then for any codeword $\mathbf{b}$ in $\mathcal{C}$, $\mathbf{b}\preceq \mathbf{c}$ implies that $\mathbf{b}=a\mathbf{c}$ for some $a\in \mathbb{F}_q$.
  For an arbitrary linear code $\mathcal{C}$, it is  hard  to determine the set of its minimal codewords.


%
%

If every codeword in  $\mathcal{C}$ is minimal, then $\mathcal{C}$ is said to be a \emph{minimal linear code}. Minimal linear codes have interesting applications in secret sharing \cite{CDY2005, CCP2014, DY2003, M1995, YD2006} and secure two-party computation \cite{ABCH1995, CMP2013}, and could be decoded with a minimum distance decoding method \cite{AB1998}. Searching for minimal linear codes has been an interesting research topic in coding theory and cryptography.

The \emph{Hamming weight} of a vector $\mathbf{v}$ is wt$(\mathbf{v}):=\# \rm{Suppt}(\mathbf{v})$. In \cite{AB1998}, Ashikhmin and Barg gave a sufficient condition  on the minimum and maximum nonzero Hamming weights for a linear code to be minimal:
\begin{lem}\label{Ashikhmin-Barg}{$($\rm {Ashikhmin-Barg} \cite{AB1998})}
A linear code $\mathcal{C}$ over $\mathbb{F}_q$ is minimal if
$$\frac{w_{\rm min}}{w_{\rm max}}>\frac{q-1}{q},$$
where $w_{\rm min}$ and $w_{\rm max}$ denote the minimum and maximum nonzero Hamming weights in the code $\mathcal{C}$, respectively.
\end{lem}

Inspired by Ding's work \cite{D2015,D2016},  many minimal linear codes with $\frac{w_{\rm min}}{w_{\rm max}}>\frac{q-1}{q}$  have been
constructed by selecting the proper defining sets or from functions over finite fields (see \cite{DD2015,HY2016,LCXM2018,SLP2016,SGP2017,TLQZ2016,WDX2015,X2016,YY2017,ZLFH2016}). Cohen et al. \cite{CMP2013} provided an example to show that the condition $\frac{w_{\rm min}}{w_{\rm max}}>\frac{q-1}{q}$ in Lemma \ref{Ashikhmin-Barg} is not necessary for a linear code to be minimal.  Recently, Ding, Heng and Zhou \cite{DHZ2018, HDZ2018}  generalized this  sufficient condition and derived a sufficient and necessary condition on all Hamming weights for a given linear code to be
minimal:
\begin{lem}\label{Heng-Ding-Zhou}{$($\rm { Ding-Heng-Zhou}\cite{DHZ2018,HDZ2018})}
A linear code $\mathcal{C}$ over $\mathbb{F}_q$ is minimal if and only if
$$\sum_{c\in\mathbb{F}_q^{*}}\rm{wt}(\mathbf{a}+c\mathbf{b})\neq (q-1)\rm{wt}(\mathbf{a})-\rm{wt}(\mathbf{b})$$
for any $\mathbb{F}_q$-linearly independent codewords $\mathbf{a},\mathbf{b}\in \mathcal{C}$.
\end{lem}

Based on {\bf Lemma \ref{Heng-Ding-Zhou}}, Ding et al. presented three infinite families of minimal binary linear codes with $\frac{w_{\rm min}}{w_{\rm max}}\leq\frac{1}{2}$ in \cite{DHZ2018} and an infinite family of minimal ternary linear codes with $\frac{w_{\rm min}}{w_{\rm max}}<\frac{2}{3}$ in \cite{HDZ2018}, respectively. In \cite{ZYW2018}, Zhang et al. constructed four families of minimal binary linear codes
with $\frac{w_{\rm min}}{w_{\rm max}}\leq\frac{1}{2}$ from Krawtchouk polynomials. Very recently, Bartoli and
Bonini \cite{BB2019} provided an infinite family of minimal linear codes; also in \cite{XQ2019} Xu and Qu constructed three classes  of minimal linear codes with $\frac{w_{\rm min}}{w_{\rm max}}<\frac{p-1}{p}$ for any odd prime $p$.


The first objective in this paper is to propose a new sufficient and necessary condition for a given linear code to be
minimal. Based on our  sufficient and necessary condition,  it is easy to construct minimal linear codes or to prove that some linear codes  are minimal. Since our sufficient and necessary  condition doesn't need the weight distribution of the linear code, the minimum (Hamming) distance is not discussed in our paper. In fact, in \cite{DHZ2018}, Ding et al. pointed out that linear codes employed for secret sharing are preferred to be minimal, in order to make the access structure of the secret sharing scheme to be special \cite{DY2003, YD2006}. Such codes may not have very good error-correcting capability. So  minimum (Hamming) distance seems not so important on these application  scenarios.

 Let $k\leq n$ be two positive integers and $q$ a prime power.   The basic question in minimal linear codes is to determine if there exists an  $[n,k]_q$ minimal linear code.
The second objective of this paper is to give the following two theorems and a corollary, which can partially answer the above question.
\begin{thm}\label{main thoerem}
Let $k$ be a positive integer  and $q$ a prime power. Then there is a positive integer $n(k;q)$ satisfies the following condition: for any positive integer $n$, there exists an $[n,k]_q$ minimal linear code if and only if $n\geq n(k;q)$.
\end{thm}

In the following theorem, we give a upper bound and a lower bound of $n(k;q)$.
\begin{thm}
$$q(k-1)<n(k;q)\leq (q-1)\frac{k(k-1)}{2}+k.$$
As a special case,
$$n(2;q)=q+1.$$
\end{thm}

As a corollary, we have:
\begin{cor}
Let $k\leq n$ be two positive integers and $q$ a prime power. Then we have
\begin{enumerate}
  \item if $n\geq (q-1)\frac{k(k-1)}{2}+k$, then there exists an $[n,k]_q$ minimal linear code;
  \item if $n\leq q(k-1)$, then any $[n,k]_q$  linear code is not minimal.
\end{enumerate}
\end{cor}

The third objective of this paper  is to present four classes of  minimal linear codes, which generalize the results about the binary case given in \cite{LY2019}. One can find that our method is much easier and more effective.

The rest of this paper is organized as follows. In Section \ref{section Preliminaries}, we give  basic results on linear algebra and  linear codes. In Section \ref{section sn}, we  present  new sufficient and necessary conditions for a codeword to be minimal and for a linear code to be minimal. In Section \ref{section Parameters}, we discuss the parameters $n, k, q$ of linear codes, and we partially answer the question  that if there exists  $[n,k]_q$ minimal linear codes.
In Section \ref{section Applications}, we use the results we obtained in this paper to generalize the four classes of minimal linear codes in \cite{LY2019}, and we give a very brief proof to show the codes are minimal. In section \ref{section Concluding remarks}, we conclude this paper.


\section{\bf Preliminaries}\label{section Preliminaries}
\subsection{Inner product and dual}
 Let $k$ be a positive integer. For two vectors $\mathbf{x}=(x_1,x_2,...,x_k)$, $\mathbf{y}=(y_1,y_2,...,y_k)\in\ \mathbb{F}_q^k$, their \emph{Euclidean inner product} is:
$$<\mathbf{x},\mathbf{y}>:=\mathbf{x}\mathbf{y}^T=\sum_{i=1}^k{x}_i{y}_i.$$

For any subset $S\subseteq \mathbb{F}_q^k$,  we define
$$S^\perp:=\{\mathbf{y}\in \mathbb{F}_q^k\mid <\mathbf{y},\mathbf{x}>=0, \it{{\rm{for\  any}}\  \mathbf{x}\in S}\}.$$
By the definition of $S^\perp$, the following two facts are immediate:\\
1. $S\subseteq (S^\perp)^\perp$;\\
2. If $S$ is a linear subspace of $\mathbb{F}_q^k$, then ${\rm{dim}}S+{\rm{dim}}S^\perp=k$.

\subsection{A general representation of linear codes}
Let $\mathcal{C}$ be an arbitrary $[n,k]_q$ linear code and $G$ its generator matrix. Then $G$ is a $k\times n$ matrix   over $\mathbb{F}_q$, rank$(G)=k$ and $\mathcal{C}$ can be expressed in the following form:
$$\mathcal{C}=\{\mathbf{c}\mathbf{(x)}=\mathbf{x}G, \mathbf{x}\in \mathbb{F}_q^k\}.$$

Let $G=(\mathbf{d}_{1}^{T},...,\mathbf{d}_{n}^{T})$, where $\mathbf{d}_1,...,\mathbf{d}_n\in \mathbb{F}_q^k$ and $D:=\{\mathbf{d}_1,...,\mathbf{d}_n\}$ be   a multiset. Then $\mathcal{C}$ can also be expressed in the following way:
$$\mathcal{C}=\mathcal{C}(D)=\{{\mathbf{c}\mathbf{(x)}}=\mathbf{c}(\mathbf{x};D)=(\mathbf{x}\mathbf{d}_{1}^{T},...,\mathbf{x}\mathbf{d}_{n}^{T}), \mathbf{x}\in \mathbb{F}_q^k\}.$$

\subsection{Properties of covering}
\begin{lem}\label{cover}
Let $\mathbf{u}=(u_1,...u_n), \mathbf{v}=(v_1,...v_n)\in \mathbb{F}_q^n$. Then the following conditions are equivalent:\\
(1) $\mathbf{u}\preceq \mathbf{v}$   $($i.e. $\rm{Suppt}(\mathbf{u})\subseteq \rm{Suppt}(\mathbf{v}))$;\\
(2) for any $1\leq i\leq k,$ if $u_i\neq 0$, then $v_i\neq 0$;\\
(3) for any $1\leq i\leq k,$ if $v_i= 0$, then $u_i= 0$;\\
(4) \rm{Zero}$(\mathbf{v})\subseteq  \rm{Zero}(\mathbf{u})$, where \rm{Zero}$(\mathbf{u}):=\{1\leq i\leq n: u_i=0\}=[1,...,n]\setminus \rm{Suppt}(\mathbf{u}).$
\end{lem}

\section{\bf A new sufficient and necessary condition for $q$-ary linear codes to be minimal}\label{section sn}
Let $k\leq n$ be two positive integers and $q$ a prime power.
Let $D$ be a multiset with elements in $\mathbb{F}_q^k$, rank$D=k$ and $\# D=n$. Let
$$\mathcal{C}(D)=\{{\mathbf{c}\mathbf{(x)}}=\mathbf{c}(\mathbf{x};D)=(\mathbf{x}\mathbf{d}_{1}^{T},...,\mathbf{x}\mathbf{d}_{n}^{T}), \mathbf{x}\in \mathbb{F}_q^k\}.$$
Then $\mathcal{C}(D)$ is an $[n,k]_q$ linear code.

Let $\mathbf{y}\in \mathbb{F}_q^k$ be an arbitrary vector. We shall first present a new sufficient and necessary condition for
the codeword $\mathbf{c(y)}\in \mathcal{C}(D)$ to be minimal. Some concepts are needed.

For any $\mathbf{y}\in \mathbb{F}_q^k$, we define
$$H(\mathbf{y}):=\mathbf{y}^\perp=\{\mathbf{x}\in \mathbb{F}_q^k\mid\mathbf{xy}^{T}=0\},$$
$$H(\mathbf{y},D):=D\cap H(\mathbf{y})=\{\mathbf{x}\in D\mid\mathbf{xy}^{T}=0\},$$
$$V(\mathbf{y},D):={\rm{Span}}(H(\mathbf{y},D)).$$
It is obvious that $H(\mathbf{y},D)\subseteq V(\mathbf{y},D)\subseteq H(\mathbf{y})$.

The following proposition is also needed.
\begin{prop}\label{31}
For any $\mathbf{x}, \mathbf{y}\in \mathbb{F}_q^k,\ \mathbf{c(x)}\preceq \mathbf{c(y)}$ if and only if $H(\mathbf{y},D)\subseteq H(\mathbf{x},D)$.\end{prop}
\begin{proof}
By the definition of the linear code, we know that
$$\mathbf{c}(\mathbf{x})=(\mathbf{x}\mathbf{d}_{1}^{T},...,\mathbf{x}\mathbf{d}_{n}^{T})\ {\rm{and}}\ \mathbf{c}(\mathbf{y})=(\mathbf{y}\mathbf{d}_{1}^{T},...,\mathbf{y}\mathbf{d}_{n}^{T}).$$
It is easy to see $i\in \rm{Zero}(\mathbf{c(x)})$ if and only if $\mathbf{d}_i\in H(\mathbf{x},D)$. Similarly, $i\in \rm{Zero}(\mathbf{c(y)})$ if and only if $\mathbf{d}_i\in H(\mathbf{y},D)$. So $\rm{Zero}(\mathbf{c(y)})\subseteq \rm{Zero}(\mathbf{c(x)})$ if and only if $H(\mathbf{y},D)\subseteq H(\mathbf{x},D)$. By {\bf Lemma \ref{cover}}, we know  $\rm{Zero}(\mathbf{c(y)})\subseteq \rm{Zero}(\mathbf{c(x)})$ if and only if $\mathbf{c(x)}\preceq \mathbf{(c(y))}$, the result is immediate.
\qed\end{proof}

Now we can give the new sufficient and necessary condition for a codeword $\mathbf{c(y)}\in \mathcal{C}(D)$ to be minimal:

\begin{thm}\label{sn}
Let $\mathbf{y}\in \mathbb{F}_q^k\backslash \{\mathbf{0}\}$. Then the following three conditions are equivalent:
\begin{enumerate}
  \item  $\mathbf{c(y)}$ is minimal in $\mathcal{C}(D)$;
  \item {\rm{dim}}$V(\mathbf{y},D)=k-1$;
  \item $V(\mathbf{y},D)=H(\mathbf{y})$.
\end{enumerate}

\end{thm}
\begin{proof}
The equivalence of (2) and (3) is easy. Since $V(\mathbf{y},D)\subseteq H(\mathbf{y})$ and dim $H(\mathbf{y})=k-1$, we get dim$V(\mathbf{y},D)=k-1$ if and only if $V(\mathbf{y},D)=H(\mathbf{y})$.

Next, we prove the equivalence of (1) and (3).

First, we will prove that if $V(\mathbf{y},D)=H(\mathbf{y})$, then  $\mathbf{c(y)}$ is minimal. Let  $\mathbf{c(x)}\preceq \mathbf{c(y)}$. Then by {\bf Proposition \ref{31}}, $H(\mathbf{y},D)\subseteq H(\mathbf{x},D)$ and $V(\mathbf{y},D)\subseteq V(\mathbf{x},D)$. Moreover, $H(\mathbf{y})=V(\mathbf{y},D)\subseteq V(\mathbf{x},D)\subseteq H(\mathbf{x})$. Since dim$H(\mathbf{y})=k-1=$ dim$H(\mathbf{x})$, we get $H(\mathbf{y})=H(\mathbf{x})$. So $\mathbf{x}\in H(\mathbf{x})^\perp=H(\mathbf{y})^\perp= \mathbb{F}_q\mathbf{y}$, and then there exists $a\in  \mathbb{F}_q$  such that $\mathbf{x}=a\mathbf{y}$ and $\mathbf{c(x)}=a\mathbf{c(y)}$. Hence $\mathbf{c(y)}$ is minimal.

Next, we will prove that if $V(\mathbf{y},D)\neq H(\mathbf{y})$, then $\mathbf{c(y)}$ is not minimal. Since $V(\mathbf{y},D)\subseteq H(\mathbf{y})$ and $V(\mathbf{y},D)\neq H(\mathbf{y})$, we get dim$V(\mathbf{y},D)<$
dim$H(\mathbf{y})=k-1$ and dim$V(\mathbf{y},D)^\perp =k-$dim$V(\mathbf{y},D)\geq 2$. It follows that there exists $\mathbf{x}\in V(\mathbf{y},D)^\perp$ satisfies that $\mathbf{x},\mathbf{y}$ are linearly independent. For any $\mathbf{d}_i\in H(\mathbf{y},D)\subseteq V(\mathbf{y},D)$, since $\mathbf{x}\in V(\mathbf{y},D)^\perp$, we have $<\mathbf{x},\mathbf{d}_i>=0$ and then $\mathbf{d}_i\in H(\mathbf{x},D).$ Therefore $H(\mathbf{y},D)\subseteq H(\mathbf{x},D)$. By {\bf Proposition \ref{31}}, we get $\mathbf{c(x)}\preceq \mathbf{c(y)}$. But $\mathbf{c(x)}, \mathbf{c(y)}$ are linearly independent.  Hence, $\mathbf{c(y)}$ is not minimal.

The proof is completed.
\qed\end{proof}

By {\bf Theorem \ref{sn}}, now we can present  a new sufficient and necessary condition for linear codes over $\mathbb{F}_q$ to be minimal.
\begin{thm}\label{sn1}
The following three conditions are equivalent:
\begin{enumerate}
  \item $\mathcal{C}(D)$ is minimal;
  \item for any $\mathbf{y}\in \mathbb{F}_q^k\backslash \{\mathbf{0}\}$,  {\rm{dim}}$V(\mathbf{y},D)=k-1$;
  \item for any $\mathbf{y}\in \mathbb{F}_q^k\backslash \{\mathbf{0}\}$, $V(\mathbf{y},D)=H(\mathbf{y})$.
\end{enumerate}
\end{thm}

\section{The Parameters of Minimal Linear Codes}\label{section Parameters}

Let $k\leq n$ be two positive integers and $q$ a prime power.   The basic question in minimal linear codes is to determine if there exists an $[n,k]_q$ minimal linear code. In this section, we give  two theorems and a corollary, which can
partially answer this question. First, we give two propositions.
\begin{prop}\label{41}
Let $D=\mathbb{F}_q^k$. Then $\mathcal{C}(D)$ is a $[q^k,k]_q$ minimal linear code.
\end{prop}
\begin{proof}
For any $\mathbf{y}\in \mathbb{F}_q^k\backslash \{\mathbf{0}\}$, $H(\mathbf{y},D)=D\cap H(\mathbf{y})=H(\mathbf{y})$ and $H(\mathbf{y},D)\subseteq V(\mathbf{y},D)\subseteq H(\mathbf{y})$, so $V(\mathbf{y},D)=H(\mathbf{y})$. By {\bf Theorem \ref{sn1}},
$\mathcal{C}(D)$ is  minimal.
\qed\end{proof}

\begin{prop}\label{42}
Let $D_1\subseteq D_2$ be two multisets with elements in $\mathbb{F}_q^k$. If $\mathcal{C}(D_1)$ is minimal, then $\mathcal{C}(D_2)$ is minimal.
\end{prop}
\begin{proof}
For any $\mathbf{y}\in \mathbb{F}_q^k\backslash \{\mathbf{0}\}$, since $D_1\subseteq D_2$, we have $H(\mathbf{y},D_1)\subseteq H(\mathbf{y},D_2)$ and $V(\mathbf{y},D_1)\subseteq V(\mathbf{y},D_2)\subseteq H(\mathbf{y})$. Since $\mathcal{C}(D_1)$ is minimal, by {\bf Theorem \ref{sn1}}, $V(\mathbf{y},D_1)= H(\mathbf{y})$. So $V(\mathbf{y},D_2)= H(\mathbf{y})$. Using {\bf Theorem \ref{sn1}} once again, we get  $\mathcal{C}(D_2)$ is minimal.
\qed\end{proof}

Let $N(k;q):=\{n\in \mathbb{N}^+ \mid$ there exists an $[n,k]_q$ minimal linear code$\}$. By {\bf Proposition \ref{41}}, we have $q^k\in N(k;q)$ and $N(k;q)\neq \emptyset$. Define $n(k;q):=$min$N(k;q)$. Then we have the following theorem.
\begin{thm}\label{43}
For any positive integer $n$, there exists an $[n,k]_q$ minimal linear code if and only if $n\geq n(k;q)$.

\end{thm}
\begin{proof}
If $n<n(k;q)$, by the definition of $n(k;q)$, there is no $[n,k]_q$ minimal linear code.
If  $n\geq n(k;q)$, by the definition of $n(k;q)$, there exists a multiset $D_1$ such that $\#D_1=n(k;q)$ and $\mathcal{C}(D_1)$ is an $[n(k;q),k]_q$ minimal linear code. Let $D_2$ be any multiset with elements in $\mathbb{F}_q^k$ and $\#D_2=n-n(k;q)$. Let $D=D_1\cup D_2$. Then $\#D=n$ and $D_1\subseteq D$. By {\bf Proposition \ref{42}}, $\mathcal{C}(D)$ is an $[n,k]_q$ minimal linear code.
\qed\end{proof}

 \vskip 0.6mm
From {\bf Theorem \ref{43}}, we see the importance of  $n(k;q)$ in the study of minimal linear codes.  Now we will give  an upper bound and a lower bound of $n(k;q)$.

 First, we shall give an upper bound of $n(k;q)$. Let $D:=\{\mathbf{x}\in \mathbb{F}_q^k\backslash \{\mathbf{0}\}\mid {\rm{wt}}(\mathbf{x})\leq 2\}$. By linear algebra and {\bf Theorem \ref{sn1}}, we can prove that  $\mathcal{C}(D)$ is minimal. In fact, a subset $D_0$ of $D$ is enough. Let $\mathbf{e}_1,\mathbf{e}_2,...,\mathbf{e}_k$ be   the standard basis of $\mathbb{F}_q^k$. We set
$$D'=\{\mathbf{e}_1,\mathbf{e}_2,...,\mathbf{e}_k\} \ \ {\rm{and}}\ \  D''=\{\mathbf{e}_i+a\mathbf{e}_j\mid 1\leq i<j\leq k, a\in \mathbb{F}_q^*\}. $$
we can see that $D_0:=D'\cup D''$ is a subset of $D$ and $\# D_0=(q-1)\frac{k(k-1)}{2}+k$.
\begin{prop}\label{51}
$\mathcal{C}(D_0)$ is a $[(q-1)\frac{k(k-1)}{2}+k,k]_q$ minimal linear code.

\end{prop}
\begin{proof}
For any $\mathbf{y}\in \mathbb{F}_q^k\backslash  \{\mathbf{0}\}$, there exists $i_0$, such that ${y_{i_0}}\neq 0$.  Define
$$
 \mathbf{\alpha}_i:=\left\{
            \begin{array}{ll}
             {\bf e}_i-y_{i_0}^{-1}y_{i}{\bf e}_{i_0}, & \hbox{if\ $i<i_0;$}\\
              {\bf e}_{i_0}-y_{i}^{-1}y_{i_0} {\bf e}_{i}, & \hbox{if\ $i>i_0,y_{i}\neq 0;$}\\
              {\bf e}_i,& \hbox{if\ $i>i_0,y_{i}= 0.$}
            \end{array}
          \right.
$$

It is obvious that $\mathbf{\alpha}_i\in D_0$ and $\mathbf{\alpha}_i\in \mathbf{y^\perp}$. These imply that
$$\{\mathbf{\alpha}_i, i\neq i_0\}\subseteq H(\mathbf{y},D_0)\subseteq V(\mathbf{y},D_0)\subseteq H(\mathbf{y}).$$
If $\{\mathbf{\alpha}_i, i\neq i_0\}$ is linearly independent, then dim $V(\mathbf{y},D_0)=k-1$. By Theorem \ref{sn} the codeword $\mathbf{c(y)}$ is minimal in $\mathcal{C}(D_0)$. Since $\mathbf{y}$ is an arbitrary vector in $\mathbb{F}_q^k\backslash  \{\mathbf{0}\}$, by Theorem \ref{sn1}, $\mathcal{C}(D_0)$ is  minimal.

Now we prove the linear independence of  $\{\mathbf{\alpha}_i, i\neq i_0\}$. Assume that $\sum_{i\neq i_0}k_i\mathbf{\alpha}_i=\mathbf{0},\ k_i\in \mathbb{F}_q$. Then
$$\sum_{i< i_0}k_i\mathbf{\alpha}_i+\sum_{i> i_0, y_i\neq 0}k_i\mathbf{\alpha}_i+\sum_{i> i_0, y_i= 0}k_i\mathbf{\alpha}_i=\mathbf{0},$$
$$\sum_{i< i_0}k_i(\mathbf{e}_i-y_{i_0}^{-1}y_{i}\mathbf{e}_{i_0}) +\sum_{i> i_0,y_i\neq 0 }k_i(\mathbf{e}_{i_0}-y_{i}^{-1}y_{i_0} \mathbf{e}_{i})+\sum_{i> i_0,y_i= 0}k_i{\bf{e}}_i=\mathbf{0}.$$
Let
$$
 t_i=\left\{
            \begin{array}{ll}
             k_i, & \hbox{if\ $i<i_0;$}\\
              -k_iy_i^{-1}y_{i_0}, & \hbox{if\ $i>i_0,y_{i}\neq 0;$}\\
              k_i,& \hbox{if\ $i>i_0,y_{i}= 0.$}
            \end{array}
          \right.
$$
So
$$\sum_{i< i_0}(t_i\mathbf{e}_i-t_iy_{i_0}^{-1}y_{i}\mathbf{e}_{i_0}) +\sum_{i> i_0,y_i\neq 0}(t_i\mathbf{e}_{i}-t_iy_{i_0}^{-1}y_{i} \mathbf{e}_{i_0} )+\sum_{i> i_0,y_i= 0}t_i{\bf{e}}_i=\mathbf{0},$$
$$\sum_{i\neq i_0}t_i\mathbf{e}_i+ (\sum_{i\neq i_0}-t_iy_{i_0}^{-1}y_{i})\mathbf{e}_{i_0}=\mathbf{0}.$$
By the linear independence of  $\{\mathbf{e}_i, 1\leq i\leq k\}$, we get $t_i=0$, and then $k_i=0$, where $\ 1\leq i\leq k$ and $i\neq i_0$.  So $\{\mathbf{\alpha}_i, i\neq i_0\}$ is linearly independent.  The proof is completed.
\qed\end{proof}

By the definition of $n(k;q)$, we have the following  upper bound.
\begin{cor}\label{52}
$n(k;q)\leq (q-1)\frac{k(k-1)}{2}+k.$
\end{cor}
Next we  shall give a lower bound of $n(k;q)$.
\begin{prop}\label{53}
Let $D$ be a multiset with elements in $\mathbb{F}_q^k\backslash  \{\mathbf{0}\}$. If $\mathcal{C}(D)$ is an $[n,k]_q$ minimal linear code, then $n> q(k-1)$.

\end{prop}

 \begin{proof}
 We define

$$ X=X(D):=\{(\mathbf{y,d})\mid <\mathbf{y,d}>=0, \mathbf{y}\in \mathbb{F}_q^k\backslash  \{\mathbf{0}\}, \mathbf{d}\in D\}.$$

Then
\begin{equation}\label{e1}
\#X=\sum_{(\mathbf{y,d})\in X}1=\sum_{\mathbf{d}\in D}\sum_{\mathbf{y}\in \mathbb{F}_q^k\backslash  \{\mathbf{0}\}\atop (\mathbf{y,d})\in X}1=\sum_{\mathbf{d}\in D}(q^{k-1}-1)=n(q^{k-1}-1).
 \end{equation}
In another way,
$$\#X=\sum_{(\mathbf{y,d})\in X}1=\sum_{\mathbf{y}\in \mathbb{F}_q^k\backslash  \{\mathbf{0}\}}\sum_{\mathbf{d}\in D\atop (\mathbf{y,d})\in X}1=\sum_{\mathbf{y}\in \mathbb{F}_q^k\backslash  \{\mathbf{0}\}}\#H(\mathbf{y},D).$$
Since $\mathcal{C}(D)$ is  minimal, $\#H(\mathbf{y},D)\geq k-1$. It follows that
\begin{equation}\label{e2}
\#X\geq (k-1)(q^k-1).
 \end{equation}

 Combining (\ref{e1}) and (\ref{e2}), $n(q^{k-1}-1)\geq (k-1)(q^k-1),$ and then
 \begin{eqnarray*}
n&\geq&\frac{q^k-1}{q^{k-1}-1}(k-1)\\
&=&\frac{q^k-q+q-1}{q^{k-1}-1}(k-1)\\
&=&q(k-1)+\frac{q-1}{q^{k-1}-1}(k-1)\\
&>&q(k-1).\\
 \end{eqnarray*}
\qed\end{proof}
By the definition of $n(k;q)$, we have the following  lower bound.
\begin{cor}\label{52'}
$n(k;q)>q(k-1).$
\end{cor}

Combining ${\bf Corollary\ \ref{52}}$ and ${\bf Corollary\ \ref{52'}}$, we have the following theorem.
\begin{thm}
$$q(k-1)<n(k;q)\leq (q-1)\frac{k(k-1)}{2}+k.$$
As a special case, when $k=2$, the specific expression of $n(2;q)$ is as follows:  $$n(2;q)=q+1.$$
\end{thm}

As a corollary, we have:
\begin{cor}\label{54}
Let $k\leq n$ be two positive integers and $q$ a prime power. Then we have
\begin{enumerate}
  \item if $n\geq (q-1)\frac{k(k-1)}{2}+k$, then there exists an $[n,k]_q$ minimal linear code;
  \item if $n\leq q(k-1)$, then any $[n,k]_q$  linear code is not minimal.
\end{enumerate}
\end{cor}

\section{Applications}\label{section Applications}

Let $k,t$ be two positive integers with $\frac{k}{2}<t<k$. Then $t\geq k-t+1$. We define the following six subsets of $\mathbb{F}_q^k$:
$$S:={\rm{Span}}\{\mathbf{e}_1,...,\mathbf{e}_t\}\setminus \{\mathbf{0}\},$$
$$S':={\rm{Span}}\{\mathbf{e}_{k-t+1},...,\mathbf{e}_k\}\setminus \{\mathbf{0}\},$$
$$S'':={\rm{Span}}\{\mathbf{e}_{k-t+2},...,\mathbf{e}_k\}\setminus \{\mathbf{0}\},$$
$$\Omega_1=\cup_{i=t+1}^k(\mathbf{e}_i+S),\ \Omega_2=\cup_{i=1}^{k-t}(\mathbf{e}_i+S'),\ \Omega_3=\cup_{i=1}^{k-t+1}(\mathbf{e}_i+S'').\ $$
Let
$$D_1=S\cup S'\cup\Omega_2,$$
$$D_2=S\cup S''\cup\Omega_3,$$
$$D_3=S\cup S'\cup\Omega_1\cup\Omega_2,$$
$$D_4=S\cup S'\cup\Omega_1\cup\Omega_3.$$

In this section, we use {\bf Proposition \ref{42}} and {\bf Proposition \ref{51}} to prove that the four classes of linear codes $\mathcal{C}(D_1),$ $\mathcal{C}(D_2),$ $\mathcal{C}(D_3),$ $\mathcal{C}(D_4)$ are all minimal. Our results generalized those in \cite{LY2019}.  In our results $q$ can be an arbitrary  prime power, while  in \cite{LY2019} $q$ equals 2. Moreover, our method is much simpler than their's.

\begin{thm}\label{61}
The four classes of linear codes $\mathcal{C}(D_1),$ $\mathcal{C}(D_2),$ $\mathcal{C}(D_3),$ $\mathcal{C}(D_4)$ are all minimal.

\end{thm}
\begin{proof}
Let $$D'=\{\mathbf{e}_1,\mathbf{e}_2,...,\mathbf{e}_k\}, \ \ \ \  D''=\{\mathbf{e}_i+a\mathbf{e}_j\mid 1\leq i<j\leq k, a\in \mathbb{F}_q^*\}$$ and $$D_0:=D'\cup D''.$$ By {\bf Proposition \ref{51}}, $\mathcal{C}(D_0)$ is minimal. If we prove that $D_0\subseteq D_i$ for $1\leq i\leq 4$, by {\bf Proposition \ref{42}}, $\mathcal{C}(D_i)$ is minimal. Since $D_1\subseteq D_3$ and $D_2\subseteq D_4$, it is enough to prove that $D_0\subseteq D_1$ and $D_0\subseteq D_2$.

First, we prove $D_0\subseteq D_1$. Since $\frac{k}{2}<t,\ t\geq k-t+1$.

 If $1\leq i\leq t$, $\mathbf{e}_i\in S$; if $i>t\geq k-t+1$, $\mathbf{e}_i\in S'$. So
 \begin{equation}\label{62}
 D'\subseteq S\cup S'\subseteq D_1.
 \end{equation}

 For $1\leq i<j\leq k$, if $j\leq t$, then $\mathbf{e}_i+a\mathbf{e}_j\in S$; if $j>t\geq k-t+1$ and $i\leq k-t$, then $\mathbf{e}_i+a\mathbf{e}_j\in \Omega_2$;  if $j>t\geq k-t+1$ and $i\geq k-t+1$, then $\mathbf{e}_i+a\mathbf{e}_j\in S'$. Thus \begin{equation}\label{62}
 D''\subseteq S \cup S'\cup \Omega_2\subseteq D_1.
 \end{equation}
 Combining (\ref{61}) and (\ref{62}), we get $D_0\subseteq D_1.$

 Next, we prove $D_0\subseteq D_2$.

 If $1\leq i\leq t$, $\mathbf{e}_i\in S$; if $i>t\geq k-t+1$, $\mathbf{e}_i\in S''$. So
 \begin{equation}\label{63}
 D'\subseteq S\cup S''\subseteq D_2.
 \end{equation}

 For $1\leq i<j\leq k$, if $j\leq t$, then $\mathbf{e}_i+a\mathbf{e}_j\in S$; if $j>t\geq k-t+1$ and $i\leq k-t+1$, then $\mathbf{e}_i+a\mathbf{e}_j\in \Omega_3$;  if $j>t\geq k-t+1$ and $i> k-t+1$, then $\mathbf{e}_i+a\mathbf{e}_j\in S''$. Thus \begin{equation}\label{64}
 D''\subseteq S \cup S'\cup \Omega_2\subseteq D_1.
 \end{equation}
 Combining (\ref{63}) and (\ref{64}), we get $D_0\subseteq D_2.$

Thus the four classes of linear codes  are all minimal. This completes the proof.
\qed\end{proof}

\section {Concluding remarks}\label{section Concluding remarks}

Let $k\leq n$ be two positive integers and $q$ a prime power.   The basic question in minimal linear codes is to determine if there exists an  $[n,k]_q$ minimal linear code.

 In this paper, we propose a new sufficient and necessary condition  for
 linear codes to be minimal. Using this condition, it is easy to construct minimal linear codes
or to prove  some linear codes are minimal.  This new sufficient and necessary condition is very powerful.

As one application, we use the new sufficient and necessary condition to partially give an answer to  the  basic question in minimal linear codes. We prove that, there is a positive integer $n(k;q)$ satisfies the following condition: for any positive integer $n$, there exists an $[n,k]_q$ minimal linear code if and only if $n\geq n(k;q)$. Moreover, we  obtain  an upper bound and a lower bound of $n(k;q)$.
When $k=2$ the expression of $n(2;q)$ is totally determined.
Next, we will continue to study $n(k;q)$. We look forward to getting a new upper bound or a new lower bound of  $n(k;q)$. Moreover, we hope to give the complete answer to the basic question in minimal linear codes. That is to say, we hope to get the specific expression of  $n(k;q)$ for any $k\geq 3$ and any prime power $q$.

As another application, we present four classes of minimal linear codes, which
generalize the results about the binary case given in \cite{LY2019}. One can find that our method is
much easier and more effective. Next, we will use our method to give more general constructions of minimal linear codes.



 {}
\end{document}